\documentclass[12 pt]{article}
\usepackage[ruled,vlined]{algorithm2e}
\usepackage{graphicx}
\usepackage[utf8]{inputenc}
\usepackage{amsmath,amsthm}
\usepackage{xcolor}

\SetCommentSty{mycommfont}
\SetKwInput{KwInput}{Input}                
\SetKwInput{KwOutput}{Output}              
\usepackage{mathtools}
\usepackage{amssymb}
\usepackage{latexsym}
\usepackage{indentfirst}
\usepackage{epsfig,float}
\usepackage{wrapfig,lipsum}
\usepackage{mathrsfs}

\newtheorem{thm}{Theorem}
\newtheorem{co}{Corollary}
\newtheorem{lem}{Lemma}

\newtheorem{Prop}{Proposition}

\newcommand{\1}{ \vspace{0.1cm} }
\newcommand{\td}{{\rm td}}
\newcommand{\tdc}{\chi_{\td}}
\newcommand{\pn}{{\rm pn}}

\newcommand{\cT}{{\cal T}}

\newcommand{\barG}{{\overline{G}}}

\let\oldenumerate\enumerate
\renewcommand{\enumerate}{
  \oldenumerate
  \setlength{\itemsep}{1pt}
  \setlength{\parskip}{0pt}
  \setlength{\parsep}{0pt}
}

\title{\bf Complexity of total dominator coloring in graphs}
\author{Michael A. Henning$^{1,}$\thanks{mahenning@uj.ac.z},  Kusum$^{2,}$\thanks{2018maz0011@iitrpr.ac.in}, Arti Pandey$^{2,}$\thanks{arti@iitrpr.ac.in}, Kaustav Paul$^{2,}$\thanks{kaustav.20maz0010@iitrpr.ac.in}}
\date{\small{
    $^1$Department of Mathematics and Applied Mathematics,\\ University of Johannesburg, Auckland Park 2006, South Africa.\\
    $^2$Department of Mathematics, \\Indian Institute of Technology Ropar, Punjab, India.\\%
   }}
   
\begin{document}
\maketitle

\begin{center}
\textbf{\large{Abstract}}
\end{center}
%
%
%

Let $G=(V,E)$ be a graph with no isolated vertices. A vertex $v$ totally dominate a vertex $w$ ($w \ne v$), if $v$ is adjacent to $w$. A set $D \subseteq V$ called a \emph{total dominating set} of $G$ if every vertex $v\in V$ is totally dominated by some vertex in $D$. The minimum cardinality of a total dominating set is the \emph{total domination number} of $G$ and is denoted by $\gamma_t(G)$. A \emph{total dominator coloring} of graph $G$ is a proper coloring of vertices of $G$, so that each vertex totally dominates some color class. The total dominator chromatic number $\tdc(G)$ of $G$ is the least number of colors required for a total dominator coloring of $G$. The \textsc{Total Dominator Coloring} problem is to find a total dominator coloring of $G$ using the minimum number of colors. It is known that the decision version of this problem is NP-complete for general graphs. We show that it remains NP-complete even when restricted to bipartite, planar and split graphs. We further study the \textsc{Total Dominator Coloring} problem for various graph classes, including trees, cographs and chain graphs. First, we characterize the trees having $\tdc(T)=\gamma_t(T)+1$, which completes the characterization of trees achieving all possible values of $\tdc(T)$. Also, we show that for a cograph $G$, $\tdc(G)$ can be computed in linear-time. Moreover, we show that $2 \le \tdc(G) \le 4$ for a chain graph $G$ and give characterization of chain graphs for every possible value of $\tdc(G)$ in linear-time.

\vspace*{2mm}

\section{Introduction}
\label{intro}
\noindent
Let $G=(V,E)$ be a graph, where $V = V(G)$ and $E=E(G)$ represents the set of vertices and set of edges in $G$, respectively. A vertex $v\in V$ is said to be adjacent to a vertex $u\in V$, if $uv\in E$. For $u,v \in V$, if $uv \in E$ then $u$ and $v$ are neighbours. The \emph{open neighbourhood} of a vertex $v\in V$ is the set of neighbours of $v$, denoted by $N_G(v)=\{u \mid uv \in E\}$, and the \emph{closed neighbourhood} of $v$ is the set $N_G[v] = N_G(v) \cup \{v\}$. The \emph{degree} of $v$ in $G$, $d_G(v)$ is the number of neighbours of $v$, and so $d_G(v) = \vert N_G(v) \vert$. A graph that contains no isolated vertex is said to be an \emph{isolate}-\emph{free} graph. A vertex $v \in V$ dominates all the vertices of its closed neighbourhood $N_G[v]$.

The concept of domination is very well studied in graphs. A set $D \subseteq V$ is a \emph{dominating set} of $G$ if every vertex in $V$ is dominated by some vertex of $D$. 
The \emph{domination number} of $G$, $\gamma(G)$ is the least cardinality of a dominating set of $G$. In an isolate-free graph $G$, a vertex $u$ \emph{totally dominates} a vertex $v$, $v \ne u$, if $uv \in E$. Thus, $v$ totally dominates all the vertices of its open neighbourhood $N_G(v)$. Note that every vertex dominates itself but does not totally dominate itself. A \emph{total dominating set}, abbreviated TD-set, of $G$ is a subset $D^t$ of $V$ such that every vertex of $V$ is totally dominated by some vertex of $D^t$. 
The \emph{total domination number} of $G$, $\gamma^t(G)$ is the least cardinality of a total dominating set of $G$. A TD-set of cardinality $\gamma_t(G)$ is called a $\gamma_t$-\emph{set of $G$}. For recent books on domination and total domination in graphs, we refer the reader to \cite{HaHeHe-20,HaHeHe-21,HaHeHe-22,HeYebook}.

A \emph{proper coloring} of $G$ is an assignment of colors to the vertices of $G$ such that if two vertices are adjacent, then they must be assigned different colors. The minimum number of colors required for a proper coloring of $G$ is the \emph{chromatic number} of $G$ and is denoted by $\chi(G)$. A subset of vertices that are assigned the same color in a proper coloring is a \emph{color class} of the coloring. A proper coloring of $G$ is said to be \emph{dominator coloring} it it also satisy the following property:  every vertex of $G$ dominates all the vertices of at least one color class. In other words, each vertex of $G$ belongs to either a singleton color class or is adjacent to every vertex of some other color class. The minimum number of colors required for a dominator coloring of $G$ is called the \emph{dominator chromatic number} of $G$, and is denoted by  $\chi_{d}(G)$. 
If a dominator coloring of $G$ uses exactly $\chi_{d}(G)$ colors, then it is called a $\chi_{d}$-\emph{coloring} of $G$.  The \textsc{Dominator Coloring} problem is to find a dominator coloring of $G$ using $\chi_{d}(G)$ colors. The decision version of the \textsc{Dominator Coloring} problem, abbreviated as the DCD problem, takes an isolate-free graph $G$ and a positive integer $k$ as input and asks whether $G$ has a dominator coloring using at most $k$ colors.

The total version of dominator coloring is also well studied in the literature. A \emph{total dominator coloring}, abbreviated TD-\emph{coloring}, of an isolate-free graph $G$ is a proper coloring of $G$ with the following additional property: each vertex of $G$ is adjacent to every vertex of some color class different from its own. The \emph{total dominator chromatic number} of $G$, denoted by $\tdc(G)$, is the minimum integer $k$ for which $G$ has a TD-coloring with $k$ colors. We have adopted the notation in the book chapter~\cite{He-21} on domination and total dominator coloring in graphs, but we remark that $\tdc(G)$ is also denoted by $\chi_{d}^t(G)$ in the literature (see, for example,~\cite{TDC2}). A TD-coloring of $G$ that uses exactly $\chi_{\td}(G)$ colors is called $\chi_{\td}$-\emph{coloring} of $G$. The \textsc{Total Dominator Coloring} problem is to find a TD-coloring of $G$ using the minimum number of colors, that is, to find a $\chi_{\td}$-coloring of $G$.

We note that a TD-coloring is only defined for isolate-free graphs. So, the graphs considered throughout this paper are isolate-free graphs. The decision version of the \textsc{Total Dominator Coloring} problem, abbreviated as the TDCD problem, takes an isolate-free graph $G$ and a positive integer $k$ as input and asks whether $G$ has a TD-coloring using at most $k$ colors. In $2009$ ~\cite{TDC11}, first introduced the concept of TD-coloring and is then extensively studied in last decade, see~\cite{TDC7,TDC8,TDC10,TDC12,TDC1,TDC9,TDC5,TDC2,TDC4,Vi12,TDC3,TDC6} and elsewhere. It is known that the TDCD problem is NP-complete for general graphs~\cite{TDC2}. The following result regarding bounds on $\tdc(G)$ is already known.
\begin{thm}{\rm \cite{TDC2,Vi12}}
\label{t:thmKa5}
For an isolate-free graph $G$, 
\[
\max \{ \gamma_t(G), \chi(G) \} \le \tdc(G) \le \gamma_t(G) + \chi(G).
\]
\end{thm}
In Theorem~\ref{t:thmKa5}, if $G$ is a bipartite graph, then $\gamma_t(G) \le \tdc(G) \le \gamma_t(G) + 2$. Both the bounds are tight for bipartite graphs as well as for trees and paths~\cite{TDC1}. Total dominator coloring of various graph classes, including paths, wheel, trees and caterpillars, are studied in~\cite{TDC2,Vi12,TDC3,TDC6}. The total dominator coloring problem is also studied on product graphs and Mycielskian graphs~\cite{TDC5,TDC4}. Further, this problem has been studied on finding the bounds and exact values of $\tdc(G)$ for some graph classes and graph operations~\cite{TDC7,TDC8,TDC10,TDC12,TDC9}. For any arbitrary tree $T$, $\gamma_t(T) \le \tdc(T) \le \gamma_t(T)+2$ and trees having $\tdc(T)=\gamma_t(T)$ are characterized in~\cite{TDC1}. The characterization of trees having $\tdc(T)=\gamma_t(T)+1$ was posed as an open problem in~\cite{TDC1}.

In this paper, we work on the complexity of the \textsc{Total Dominator Coloring} problem for some graph classes, namely chain graphs, cographs, bipartite graphs, planar graphs and split graphs. First, we give a characterization of trees having $\tdc(T)=\gamma_t(T)+1$, that completes the characterization of trees for every possible value of $\tdc(T)$. We remark that the condition given in this characterization cannot be checked in polynomial time. Then, we compute the value of the total dominator chromatic number for both connected and disconnected cographs in linear-time. Next, we show that for a chain graph $G$, $2 \le \tdc(G) \le 4$ and characterize the class of chain graphs for every possible value of $\tdc(G)$ in linear-time. On the other hand, to the best of our knowledge, there is only one hardness result known for the TDCD problem, which states that the TDCD problem is NP-complete for general graphs~\cite{TDC2}. We extend the study of the \textsc{Total Dominator Coloring} problem in this direction by showing that the TDCD problem remains NP-complete even when restricted to planar graphs, connected bipartite graphs and split graphs. This also shows that the TDCD problem remains NP-complete for chordal graphs, as split graphs is a subclass of chordal graphs.

This paper is organised as follows. In Section~2, we define the necessary graph theory notation and mentioned some known results that will be used throughout the paper. In Section~3, we focus on the TD-coloring of trees and characterize the trees $T$ having $\tdc(T)=\gamma_t(T)+1$. In Section~4, we compute $\tdc(G)$ for any cograph $G$. In Section~5, we investigate the dominator coloring and the TD-coloring of chain graphs. In Section~6, we show that the TDCD problem is NP-complete even for planar graphs, connected bipartite graphs and split graphs. Finally, Section~7 concludes the paper.

\section{Preliminaries}
\noindent
Let $G=(V,E)$ be a graph. We assume that all the graphs considered in this paper are simple, non-trivial, isolate-free and undirected. A connected acyclic graph is called a \emph{tree}. In a tree $T$, a degree~$1$ vertex is called a \emph{leaf} and its neighbour a \emph{support vertex}. Let $P_n$ denote the path on $n$ vertices. A graph $G$ is a \emph{cograph} if $P_4$ is not present as an induced subgraph of $G$, that is, $G$ is $P_4$-free. Equivalently~\cite{cograph}, a graph $G$ of order at least~$2$ is a cograph if and only if $G$ or its complement $\barG$ is not connected.

If $D$ is a minimal TD-set of $G$, then the \emph{$D$-private neighbourhood} of a vertex $u \in D$ is the set of vertices that are totally dominated by $u$ only, and is denoted by $\pn(u,D)$. Thus, if $w \in \pn(u,D)$, then $N(w) \cap D = \{u\}$. If $\pn(u,D)=\{w\}$, then $w$ is the only vertex in the $D$-private neighbourhood of~$u$. We define the sets $D_I = \{u\in D \colon \vert \pn(u,D) \vert=1 \}$ and $D_R=D \setminus D_I$. Thus, $D=D_{I}\cup D_{R}$.

An \emph{optimal TD-coloring} of $G$ is a $\chi_{\td}$-coloring of $G$. Let $\cal{H}$ be a $\chi_{\td}$-coloring of $G$. The color class $V_i^{\cal{H}}$ is the set of vertices receiving color $i$ in $\cal{H}$, where $1 \le i \le \tdc(G)$. Let $C^{\cal{H}}=\{ V_1^{ \cal{H}}, V_2^{ \cal{H}}, \ldots, V_{\tdc(G)}^{\cal{H}}\}$ be the collection of color classes of $\cal{H}$. If $ \vert V_i^{\cal{H}} \vert=1$ for a color class $V_i^{\cal{H}}$, then $V_i^{\cal{H}}$ is called a \emph{solitary color class} and the vertex $v\in V_i^{\cal{H}}$ is called a \emph{solitary vertex}. A color class $V_i^{\cal{H}}$ is said to be a \emph{free color class} if every vertex of $G$ totally dominates a color class other than $V_i^{\cal{H}}$. Let 
\begin{enumerate}
\item[$\bullet$] $C_0^{\cal{H}}$ be a minimum cardinality subset of $C^{\cal{H}}$ such that each $u \in V$ totally dominates some color class of $C_0^{\cal{H}}$,  \1
\item[$\bullet$] $C_P^{\cal{H}}$ be the subset of $C^{\cal{H}}$ such that each color class $R \in C_P^{\cal{H}}$ is a solitary color class,  \1
\item[$\bullet$] $C_S^{\cal{H}}$ be the subset of $C^{\cal{H}}$ such that each color class $R \in C_S^{\cal{H}}$ contains more than one vertex and is totally dominated by some vertex of $G$, and \1
\item[$\bullet$] $C_G^{\cal{H}}$ be the subset of $C^{\cal{H}}$ such that each color class $R \in C_G^{\cal{H}}$ contains more than one vertex and is not totally dominated by any vertex of $G$.
\end{enumerate}
The sets $C_P^{ \cal{H}}$, $C_S^{ \cal{H}}$ and $C_G^{ \cal{H}}$ forms a partition of the color classes of $\cal{H}$ and thus, $C^{ \cal{H}}= C_P^{ \cal{H}} \cup C_S^{ \cal{H}} \cup C_G^{ \cal{H}}$. Let $A^{ \cal{H}}$ be the set of solitary vertices in the coloring $\cal{H}$, and let $B^{ \cal{H}}$ be the set of all the vertices in color classes of $C_G^{ \cal{H}}$. Thus,
\begin{enumerate}
\item[$\bullet$] $A^{ \cal{H}}=\{u \in R \colon R \in C_P^{ \cal{H}}\}$,  and \1
\item[$\bullet$] $B^{ \cal{H}}=\{u \in R \colon R \in C_G^{ \cal{H}}\}$.
\end{enumerate}

Let $D_0^{\cal{H}}$ be the set constructed by picking exactly one vertex from each color class of $C_0^{\cal{H}}$. Also, let $D_S^{\cal{H}}$ be the set constructed by picking one vertex from each color class of $C_S^{\cal{H}}$. We note that $\vert C_{S}^{\cal{H}}\vert= \vert D_S^{\cal{H}} \vert$.

\begin{figure}[htb]
\begin{center}
\includegraphics[width=11cm, height=1.2cm]{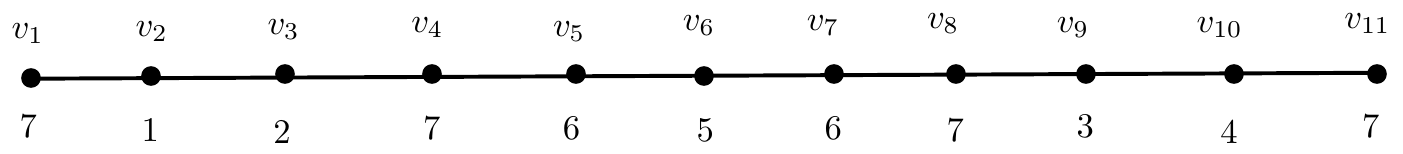}
\caption{A tree $T = P_{11}$ and a $\chi_{\td}$-coloring $\cal{H}$ of $T$}
\label{fig:1}
\end{center}
\end{figure}

We now illustrate the above definitions with an example. Let $T = P_{11}$ be the path $v_1v_2 \ldots v_{11}$ of order~$11$, and let $\cal{H}$ be the $\chi_{\td}$-coloring of $T$ given in Fig. 1. In this example, the following properties hold in the tree $T$.
\begin{enumerate}
\item[$\bullet$] $\tdc(T)=7$, \1
\item[$\bullet$] $V_1^{\cal{H}}=\{v_2\}$, $V_2^{\cal{H}}=\{v_3\}$, $V_3^{\cal{H}}=\{v_9\}$, $V_4^{\cal{H}}=\{v_{10}\}$, $V_5^{\cal{H}}=\{v_6\}$, $V_6^{\cal{H}}=\{v_5, v_7\}$, and $V_7^{\cal{H}}=\{v_1, v_4, v_8, v_{11}\}$ are the color classes, \1
\item[$\bullet$] $C^{\cal{H}}=\{V_1^{\cal{H}}, V_2^{\cal{H}} ,V_3^{\cal{H}},V_4^{\cal{H}}, V_5^{\cal{H}}, V_6^{\cal{H}},V_7^{\cal{H}}\}$, \1
\item[$\bullet$] $C_0^{\cal{H}} = \{V_1^{\cal{H}}, V_2^{\cal{H}}, V_3^{\cal{H}},V_4^{\cal{H}}, V_5^{\cal{H}}, V_6^{\cal{H}}\}$, \1
\item[$\bullet$] The solitary vertices are $v_2,v_3,v_6,v_9,v_{10}$, \1
\item[$\bullet$] $C_P^{\cal{H}}=\{V_1^{\cal{H}}, V_2^{\cal{H}} ,V_3^{\cal{H}},V_4^{\cal{H}}, V_5^{\cal{H}}\}$, \1
\item[$\bullet$] $C_S^{\cal{H}}=\{V_6^{\cal{H}}\}$, \1
\item[$\bullet$] $C_G^{\cal{H}}=\{V_7^{\cal{H}}\}$, \1
\item[$\bullet$] $A^{\cal{H}}=\{v_2,v_3,v_6,v_9,v_{10}\}$, \1
\item[$\bullet$] $B^{\cal{H}}=\{v_1,v_4,v_8,v_{11}\}$, \1
\item[$\bullet$] $D_S^{\cal{H}}=\{v_5\}$, \1
\item[$\bullet$] $D_0^{\cal{H}}=\{v_2,v_3,v_5,v_6,v_9,v_{10}\}$, \1
\item[$\bullet$] For the dominating set $D_0^{\cal{H}}$, $D_I=\{v_3,v_5\}$, and $D_R=\{v_2,v_6,v_9,v_{10}\}$.
\end{enumerate}

For planar graph $G$, $\chi(G) = 4$. Using Theorem~\ref{t:thmKa5}, we have $\gamma_t(G) \le \tdc(G) \le \gamma_t(G) + 4$. Now, we formally state the results regarding bounds on $\tdc(G)$ for bipartite and planar graphs.
\begin{co}
\label{planarbound}
The following properties hold. 
\begin{enumerate}
\item[{\rm (a)}] If $G$ is a bipartite graph, then $\gamma_t(G) \le \tdc(G) \le \gamma_t(G) + 2$.
\item[{\rm (b)}] If $G$ is a planar graph, then $\gamma_t(G) \le \tdc(G) \le \gamma_t(G) + 4$.
\end{enumerate}
\end{co}

Since trees are a subclass of bipartite graph, the bounds in Corollary~\ref{planarbound}(a) hold if $G$ is a tree. We note that both the bounds in Corollary~\ref{planarbound}(a) are achievable for bipartite graphs as well as for trees.

%
%

For the characterization of trees having $\tdc(T) =\gamma_t(T)$ given in \cite{TDC1}, a family of trees $\cT$ is constructed as: $\cT=P_2 \cup \{\text{trees obtained by}$ connecting $k \geq 1$ disjoint stars of order at least three using $(k-1)$ edges joining leaf vertices such that the center of each original star remains a stem$\}$. The following results are known for trees.
\begin{thm}{\rm \cite{TDC1}}
\label{t:tree-1}
For a tree $T$, $\gamma_t(T) = \tdc(T)$ if and only if $T \in \cT$.
\end{thm}

\begin{thm}{\rm \cite{TDC1}}
\label{t:tree-2}
For a tree $T \notin \cT$, the following statements holds:
\begin{enumerate}
\item[{\rm (a)}] If $\tdc(T) = \gamma_t(T) + 1$, then $T$ admits a $TDC$ using $\chi_d^t(T)$ colors having a free color class.
\item[{\rm (b)}] If $\tdc(T) = \gamma_t(T) + 2$, then $T$ admits a $TDC$ using $\chi_d^t(T)$ colors having two free color classes.
\end{enumerate}
\end{thm}

\section{Characterization of trees $T$ having $\tdc(T)=\gamma^{t}(T)+1$}
\noindent
Throughout this section, we assume that $T=(V,E)$ is a non-trivial tree. By Corollary~\ref{planarbound}(a), $\tdc(T)$ takes one of the three values $\gamma^{t}(T)$, $\gamma^{t}(T)+1$ or $\gamma^{t}(T)+2$. Further, it is shown in~\cite{TDC1} that there are infinitely many trees for each value of $\tdc(T)$. Also, recall that Theorem~\ref{t:tree-1} gives a characterization of the trees $T$ satisfying $\gamma_t(T) = \tdc(T)$. In this section, we characterize trees $T$ satisfying $\gamma_t(T) = \tdc(T) + 1$, thereby completing a characterization of trees $T$ having every possible value of $\tdc(T)$.

We first prove properties of $\tdc$-colorings of a tree.
\begin{Prop}
\label{P1}
Every support vertex in any $\tdc$-coloring of a tree $T$ is solitary.
\end{Prop}
\begin{proof}
Let $\cal{H}$ be a $TD$-coloring of a tree $T$ using $\chi_d^t(T)$ colors. On the contrary, assume that there exists a stem $v$ which is not solitary. Since $v$ is a stem, there must be a leaf vertex $x$ adjacent to this stem, which is not adjacent to any other vertex of $T$. In any $TD$-coloring of $T$, each vertex of $T$ is properly colored and totally dominates some color class. Let $v$ belongs to color class $R$, which is not solitary. Then, the vertex $x$ is not totally dominating any color class, which is a contradiction to the fact that $\mathcal{H}$ is a $TD$-coloring  of $T$ using $\tdc$ colors. Hence, the result follows.
%
\end{proof}

In the next result, we consider the trees having at least three vertices and we establish the existence of an optimal TD-coloring such that leaves that are adjacent to same support vertex can be given same color.
\begin{Prop}
\label{P2}
If $T$ is a tree of order~$n \ge 3$, then there exists a $\tdc$-coloring of $T$ such that all leaf neighbours of a support vertex belong to the same color class.
\end{Prop}
\begin{proof}
Among all $\tdc$-colorings of the tree $T$, let $\cal{H}$ be chosen so that the number of support vertices in $T$ whose leaf neighbours are not all colored with the same color is minimum. Let $u$ be an arbitrary support vertex of $T$, and let $u$ have color~$1$ in the coloring $\cal{H}$. Further, let ${\cal{H}}_1$ be the color class that contains~$u$. Suppose that the leaf neighbours of $u$ does not belong to the same color class. By Proposition~\ref{P1}, ${\cal{H}}_1$ is a solitary color class, and so ${\cal{H}}_1 = \{u\}$. The vertex $u$ necessarily totally dominates some color class, say ${\cal{H}}_2$ and let every vertex in ${\cal{H}}_2$ is colored with color~$2$. If some leaf neighbour of $u$ is colored~$2$, then recolor all the leaf neighbours of $u$ with color~$2$. If no leaf neighbour of $u$ is colored~$2$, then recolor all the leaf neighbours of $u$ with an existing color used to color one of the leaf neighbours of~$u$. Let ${\cal{H}}'$ be the resulting coloring of the vertices of $T$. This produces a $\tdc$-coloring of $T$ with fewer support vertices whose leaf neighbours are not all colored with the same color, contradicting our choice of the $\tdc$-coloring $\cal{H}$. Therefore, the $\tdc$-coloring $\cal{H}$ colors all leaf neighbours of a support vertex with the same color.
\end{proof}

We note that it is not necessarily true that if $T$ is a tree with~$n \ge 3$, then there exists a $\tdc$-coloring of $T$ that colors all leaves with the same color. For example, if $T$ is a path $P_6$ and ${\cal{C}}$ is a TD-coloring that colors both leaves with the same color, then an additional four colors are needed for ${\cal{C}}$ to be a TD-coloring. Such a TD-coloring, therefore uses five colors. However, $\tdc(T) = \gamma_t(T) = 4$, and so ${\cal{C}}$ is not a $\tdc$-coloring of $T$. We remark, however, that $P_6$ belongs to the tree family $\cT$ defined earlier and, by Theorem~\ref{t:tree-1}, a tree $T$ belongs to this family $\cT$ if and only if $\gamma_t(T) = \tdc(T)$. We show next that if $T$ is a tree that does not belong to the family $\cT$, then there does exist a $\tdc$-coloring of $T$ that colors all the leaves with the same color.

\begin{Prop}
\label{P3}
If $T$ is a tree and $T \notin \cT$, then there exists a $\tdc$-coloring of $T$ that colors all leaves with the same color.
\end{Prop}\vspace*{-.5 cm}
\begin{proof} 
Let $T$ be a tree that does not belong to the family~$\cT$. By Theorem~\ref{t:tree-1}, $\gamma_t(T) \ne \tdc(T)$, implying by Corollary~\ref{planarbound} that either $\tdc(T) = \gamma_t(T) + 1$ or $\tdc(T) = \gamma_t(T) + 2$. By Theorem~\ref{t:tree-2}, there exists a $\tdc$-coloring $\cal{H}$ of $T$ which contains a free color class, say $R$ where vertices in $R$ are colored using color~$r$. 

We now construct a $\tdc$-coloring of $T$ as follows. Since $T \notin \cT$, we note that the tree $T$ is not a star, implying that each support vertex of $T$ has some non-leaf neighbour. For each support vertex $u$ in $T$, do the following. If all the leaf neighbours of $u$ are colored using color~$r$, then we make no change to the colors of these leaf neighbours, and they all remain colored using color~$r$. Suppose, however, that some leaf neighbour of $u$ is not colored using color~$r$. The vertex~$u$ totally dominates some color class, say $S$, where vertices of $S$ are colored using color~$s$. 

Now, if a non-leaf neighbour of $u$ is colored with the color~$s$, then recolor all leaf neighbours of $u$ using color~$r$. Otherwise, if no non-leaf neighbour of $u$ is colored with the color~$s$, then some, but not all the leaf neighbours of $u$ are colored with the color~$s$. In this case, we select an arbitrary non-leaf neighbour of $u$ and recolor it with the color~$s$ and recolor all the leaf-neighbours of $u$ with the color~$r$. We do this for every support vertex in $T$. The resulting $\tdc$-coloring of $T$ colors all the leaves with the same color.
\end{proof}

\begin{co}
\label{CG}
If $T$ is a tree satisfying $\tdc(T) = \gamma_t(T)+1$, then there exists a $\tdc$-coloring $\cal{H}$ of $T$ that colors all the leaves with the same color and such that $\vert C_{G}^{\cal{H}}\vert = 1$.
\end{co}\vspace*{-.4 cm}
\begin{proof}
Let $T$ be a tree satisfying $\tdc(T) = \gamma_t(T)+1$. By Proposition~\ref{P3}, there exists a $\tdc$-coloring $\cal{H}$ of $T$ that colors all leaves with the same color. Assume that color class $R$ contains all the leaves of $T$. Since $T \notin \cT$, the tree $T$ is not a star. Thus, there does not exist any vertex in $T$ which is adjacent to all the leaves of $T$, implying that $R$ is a free color class of $\cal{H}$ and hence the color class $R$ belongs to the set $C_{G}^{\cal{H}}$. We show that $C_{G}^{\cal{H}} = \{R\}$. On the contrary, suppose that $\vert C_{G}^{\cal{H}}\vert \ge 2$. In this case, the set consisting of one vertex from each $Q \in C_{P}^{\cal{H}} \cup C_{S}^{\cal{H}}$ forms a TD-set of $T$, which is of cardinality $\vert C_{P}^{\cal{H}}\vert + \vert C_{S}^{\cal{H}}\vert \le \vert{\cal{H}}\vert - \vert C_{G}^{\cal{H}}\vert \le \tdc(T) - 2 = \gamma_t(T) - 1$, a contradiction. Therefore, $C_{G}^{\cal{H}} = \{R\}$, and so $\vert C_{G}^{\cal{H}}\vert = 1$.
\end{proof}

We next prove some key lemmas that we will need to prove our characterization of trees $T$ satisfying $\gamma_t(T) = \tdc(T) + 1$.
\begin{lem}
\label{t1}
If $\cal{H}$ is a $\tdc$-coloring in a tree $T$, then the following properties hold. 
\begin{enumerate}
\item[{\rm (a)}] Every $R \in C_S^{\cal{H}}$ is totally dominated by exactly one vertex.
\item[{\rm (b)}] $A^{\cal{H}} \cup D_S^{\cal{H}}$ is a TD-set of $T$.
\end{enumerate}
\end{lem}
\begin{proof}
Let $\cal{H}$ be a $\tdc$-coloring of $T$. Let $R \in C_S^{\cal{H}}$. If $R$ is totally dominated by two or more vertices, then any two such vertices, together with any two vertices from $R$, induce a subgraph of the tree $T$ that contains a $4$-cycle, which is a contradiction. Hence, $R$ is totally dominated by exactly one vertex. This proves part~(a).

To prove part~(b), 
we assume that $v$ is an arbitrary vertex of $T$. As $\cal{H}$ is a $\tdc$-coloring of $T$, there exists a color class, say $R$, such that $v$ totally dominates $R$. Thus, $R \in C_P^{\cal{H}} \cup C_S^{\cal{H}}$. If $R \in C_P^{\cal{H}}$, then $v$ is totally dominated by some vertex of $A^{\cal{H}}$. Otherwise, if $R \in C_S^{\cal{H}}$, then $v$ is totally dominated by some vertex of $D_S^{\cal{H}}$. Therefore, $A^{\cal{H}} \cup D_S^{\cal{H}}$ is a TD-set of $T$. This proves part~(b).
\end{proof}

\begin{lem}
\label{t3}
If $T$ is a tree satisfying $\tdc(T) = \gamma_t(T)+1$, then there exists a $\tdc$-coloring $\cal{H}$ of $T$ such that $A^{\cal{H}} \cup D_S^{\cal{H}}$ is a $\gamma_t$-set of $T$.
\end{lem}
\begin{proof}
Let $T$ be a tree and $\tdc(T) = \gamma_t(T)+1$. By Corollary~\ref{CG}, there exists a $\tdc$-coloring $ \cal{H} $ of $T$ satisfying $ \vert C_{G}^{\cal{H}} \vert = 1$. We note that $\tdc(T) = \vert C_P^{\cal{H}} \vert + \vert C_S^{\cal{H}} \vert + \vert C_G^{\cal{H}} \vert $.  Moreover, by definition we have $\vert C_P^{\cal{H}} \vert = \vert A^{\cal{H}} \vert $ and $\vert C_S^{\cal{H}} \vert = \vert D_S^{\cal{H}} \vert $. Thus,
\[
\begin{array}{lcl}
\gamma_t(T) & = & \tdc(T) - 1 \1 \\
& = & ( \vert C_P^{\cal{H}} \vert + \vert C_S^{\cal{H}} \vert + \vert C_G^{\cal{H}} \vert ) - 1 \1 \\
& = &  (1 + \vert A^{\cal{H}}\vert + \vert D_S^{\cal{H}} \vert ) - 1 \1 \\
& = &  \vert A^{\cal{H}} \vert + \vert D_S^{\cal{H}} \vert .
\end{array}
\]

By Lemma~\ref{t1}(b), we infer that the TD-set $A^{\cal{H}} \cup D_S^{\cal{H}}$ of $T$ is therefore a minimum TD-set, that is, $A^{\cal{H}} \cup D_S^{\cal{H}}$ is a $\gamma_t$-set of $T$.  
\end{proof}

Before presenting our main result of this section, we introduce some additional notation. Let $T$ be a non-trivial tree satisfying $T \notin \cT$, and $D$ be a $\gamma_t$-set of the tree $T$ and $S \subseteq D$, then $v \in V(T)$ is called a $(D,S)$-\emph{bad vertex} if $\vert N_T(v) \cap D \vert \ge 2$ and $N_T(v)\cap D \subseteq S$. We are now in a position to provide a characterization of trees $T$ satisfying $\gamma_t(T) = \tdc(T) + 1$.

\begin{thm}
\label{tree:char}
If $T$ is a non-trivial tree and $T \notin \cT$, then $\tdc(T) = \gamma_t(T)+1$ if and only if there exists a $\gamma_t$-set $D$ of $T$ and a partition $(D_{1},D_{2})$ of $D$ satisfying the following properties: 
\begin{enumerate}
\item[{\rm (a)}] $D_2 \subseteq D_1$,
\item[{\rm (b)}] $T$ contains no $(D_2,D)$-bad vertex, and
\item[{\rm (c)}] the set $V(T)\setminus (D_1 \cup N[S])$ is independent, where
$\displaystyle{S= \bigcup_{v \in D_2} \pn(v,D)}$.
\end{enumerate}
\end{thm}
\begin{proof}
Let $T$ be a tree satisfying $\tdc(T) = \gamma_t(T)+1$. By Corollary~\ref{CG}, there exists a $\tdc$-coloring $\cal{H}$ of $T$ that colors all leaves with the same color and such that $\vert C_{G}^{\cal{H}}\vert = 1$. Let $R$ be the set of all leaves of $T$, and so $C_{G}^{\cal{H}} =\{ R\}$. Let $D_1^{\cal{H}} \subseteq V$ which contain all the solitary vertices from the color classes of $C_P^{\cal{H}}$, and let $D_2^{\cal{H}}\subseteq V$ contains precisely one vertex from each $Q \in C_S^{\cal{H}}$. Thus, $\vert D_1^{\cal{H}}\vert = \vert C_P^{\cal{H}}\vert =  \vert A^{\cal{H}}\vert$ and $\vert D_2^{\cal{H}}\vert = \vert C_S^{\cal{H}}\vert =  \vert D_S^{\cal{H}}\vert$. By Lemma~\ref{t3} and its proof, the set $D^{\cal{H}} = D_1^{\cal{H}} \cup D_2^{\cal{H}}$ is a $\gamma_t$-set of $T$. Thus, $\vert D^{\cal{H}}\vert = \gamma_t(T)$ and $\pn(x,D^{\cal{H}}) \ne \emptyset$, for each $x \in D^{\cal{H}}$.

Let $x \in D_2^{\cal{H}}$ and let $X$ be the color class of $\cal{H}$ such that $x\in X$. Thus, $x \in X$ and $\vert X \vert \ge 2$. We show that $x$ has a unique $D^{\cal{H}}$-private neighbour, that is, $\vert \pn(x,D^{\cal{H}})\vert = 1$. As observed earlier, $\pn(x,D^{\cal{H}}) \ne \emptyset$. Let $y \in \pn(x,D^{\cal{H}})$, and thus, the only neighbour of $y$ in $D^{\cal{H}}$ is~$x$. Since $\cal{H}$ is a TD-coloring of $T$, $y$ totally dominates some color class, say $Y \in  C_P^{\cal{H}} \cup  C_S^{\cal{H}}$. If $Y$ is different from $X$ and since $R$ is a free color class, then the vertex $y$ would be adjacent to at least two vertices in the $D^{\cal{H}}$, contradicting the supposition that $y \in \pn(x,D^{\cal{H}})$. Hence, $y$ totally dominates color class $X$ and $x\in X$.

To the contrary, suppose that $\vert \pn(x,D^{\cal{H}})\vert \ge 2$, and let $z$ be a vertex in $\pn(x,D^{\cal{H}})$ different from~$y$. Analogous arguments as given for the vertex~$y$ show that $z$ totally dominates color class $X$. However, $\vert X \vert \ge 2$. Thus, any two vertices from the color class $X$, together with the vertices $y$ and $z$, induce a subgraph of the tree $T$ that contains a $4$-cycle, a contradiction. Hence, $\pn(x,D^{\cal{H}}) = \{y\}$, that is, the vertex $x$ has a unique $D^{\cal{H}}$-private neighbour, where $x$ is an arbitrary vertex in $D_2^{\cal{H}}$. Now, recall that $D_I = \{v \in D \colon \vert \pn(v,D)\vert =1 \}$, and so $D_2^{\cal{H}} \subseteq D_I$. Let
\[
S= \bigcup_{x \in D_2^{\cal{H}}} \pn(x,D^{\cal{H}}).
\]

By our earlier observations, $\vert \pn(x,D^{\cal{H}}) \vert = 1$ and the vertex in $\pn(x,D^{\cal{H}})$ totally dominates the color class containing~$x$, for every vertex $x \in D_2^{\cal{H}}$. Thus, $N_T[S]$ contains all  vertices that belong to the set $C_S^{\cal{H}}$. This implies that $V(T) \setminus (D_1^{\cal{H}} \cup N[S]) \subseteq B^{\cal{H}}$, where $B^{\cal{H}}=R$. Now, since $R$ is independent, the set $V(T) \setminus (D_1^{\cal{H}} \cup N[S])$ is also independent.

Next, we show that there is no $(D_2^{\cal{H}},D^{\cal{H}})$-bad vertex. To the contrary, suppose that there exists a $(D_2^{\cal{H}},D^{\cal{H}})$-bad vertex, say $v \in V(T)$. Thus, $\vert N_T(v) \cap D^{\cal{H}}\vert \ge 2$ and $N_T(v) \cap D^{\cal{H}} \subseteq D_2^{\cal{H}}$. Thus, $v$ cannot totally dominate any $K \in C_P^{\cal{H}}$. Further, $v$ is not a $D^{\cal{H}}$-private neighbour of any $w \in D^{\cal{H}}$. 
Let $v$ totally dominates $Q \in C_P^{\cal{H}} \cup  C_S^{\cal{H}}$. Necessarily, $Q$ belongs to $C_S^{\cal{H}}$. Let $u \in Q \cap D^{\cal{H}}$, and let $u'\in Q$ such that $u' \ne u$. By our earlier observations, $\vert\pn(u,D^{\cal{H}})\vert = 1$. Let $x \in \pn(u,D^{\cal{H}})$. Now, the set $\{u,u',v,x\}$ induces a $4$-cycle in $T$, a contradiction. 
As $Q$ was arbitrary and $v$ does not totally dominate any $Q \in C_S^{\cal{H}}$, a contradiction to $\cal{H}$ being a $TD$-coloring of $T$. Hence, there is no $(D_2^{\cal{H}},D^{\cal{H}})$-bad vertex. Thus the properties (a),~(b) and~(c) all hold, where $D_1 = D_1^{\cal{H}}$ and $D_2 = D_2^{\cal{H}}$.

Conversely, let $T$ be a non-trivial tree and $T \notin \cT$, and let there exists a $\gamma_t$-set $D$ of $T$ and a partition $(D_{1},D_{2})$ of $D$ satisfying the three properties (a),~(b) and~(c), that is, (a) $D_2 \subseteq D_I$, (b) $T$ contains no $(D_2,D)$-bad vertex, and (c) the set $V(T)\setminus (D_1 \cup N[S])$ is independent, where
\[
S = \bigcup_{v \in D_2} \pn(v,D).
\]

Let ${\cal C}$ be the coloring of the vertices of $T$ defined as follows. Color each vertex in $D$ with a unique color. Further, for each vertex $x \in D_2$ and its unique $D$-private neighbour $y \in \pn(x,D)$, we color all the vertices in $N_T(y)$ with the same color used to color~$x$. Finally, we color all the remaining uncolored vertices with one new color. Since $V(T)\setminus (D_1 \cup N[S])$ is independent and since $T$ contains no $(D_2,D)$-bad vertex, we infer that ${\cal C}$ is a TD-coloring of $T$, which implies that $\tdc(T) \le \vert {\cal C} \vert = \vert D \vert + 1 = \gamma_t(T) + 1$. However, $\tdc(T) \ge \gamma_t(T) + 1$, since by supposition $T \notin \cT$. Therefore, $\tdc(T) = \gamma_t(T) + 1$. 
\end{proof}

By Theorem~\ref{t:tree-1} and Theorem~\ref{tree:char}, we have a characterization of trees for all three possible values of the total dominator chromatic number.

\section{Total dominator chromatic number of cographs}
\noindent
In this section, we compute the total dominator chromatic number of connected and disconnected cographs in terms of the chromatic number of cographs. 
In a $TD$-coloring $\cal{H}$ of $G$, we call a color class $R \in C_0^{\cal{H}}$ as an \emph{exclusive color class} and the corresponding color an \emph{exclusive color}. The remaining colors in the coloring $\cal{H}$ we call \emph{non-exclusive colors}. The number of exclusive colors will be unique for a given TD-coloring $\cal{H}$ of $G$, but if we change the TD-coloring, then this may change accordingly.

First, we show that the total dominator chromatic number and the chromatic number coincides for connected cographs. Further, we prove that in any optimal TD-coloring of a connected cograph, there are at least two exclusive color classes.
\begin{thm}
\label{conco}
If $G$ is a connected cograph, then $\tdc(G) = \chi(G)$. Further, if $\cal{H}$ is a $\tdc$-coloring $\cal{H}$ of $G$, then $\vert C_0^{\cal{H}}\vert =2$.
\end{thm}
\begin{proof}
Let $G$ be a connected cograph. Thus, the graph $\barG$ is not connected, implying that $V(G)$ can be partitioned into two non-empty disjoint subsets $P$ and $Q$ such that every $x \in P$ is adjacent to every $y \in Q$ in the graph $G$. 
Assume that $\cal{H}$ is a proper coloring of $G$. 
Clearly, for any color class $A \in C^{\cal{H}}$, either $A \subseteq P$ or $A \subseteq Q$ not both.
Let $a$ and $b$ be the colors used to color vertices in $P$ and $Q$, respectively. Assume that $V_a^{\cal{H}}$ and $V_b^{\cal{H}}$ be the color classes of color $a$ and $b$, respectively. 
We note that $V_a^{\cal{H}} \subseteq P$ and $V_b^{\cal{H}} \subseteq Q$. Thus, each $x \in P$ totally dominates $V_b^{\cal{H}}$, and each vertex of $Q$ totally dominates $V_a^{\cal{H}}$, implying that $\cal{H}$ is a TD-coloring of $G$. Therefore, $\chi(G) \le \tdc(G) \le \vert \cal{H}\vert  = \chi(G)$. Hence, $\tdc(G) = \chi(G)$. Moreover, since each $v \in V$ either totally dominates $V_a^{\cal{H}}$ or $V_b^{\cal{H}}$, the set $C_0^{\cal{H}} = \{V_a^{\cal{H}},V_b^{\cal{H}}\}$. Thus, $\vert C_0^{\cal{H}}\vert  = 2$ and there are two exclusive colors required in an optimal TD-coloring of $G$.  
\end{proof}

Next, we consider disconnected cographs $G$. Using the property that every component of $G$ is itself a connected cograph, we provide an expression for computing $\tdc(G)$ in terms of $\chi(G)$ and the number of components of $G$.
\begin{thm}\label{disconco}
If $G$ is a disconnected cograph with $k$ components, then
\[
\tdc(G)=\chi(G)+2(k-1).
\]
\end{thm}
\begin{proof} Let $G$ be a disconnected graph with $k \ge 2$ components $G_1, \ldots, G_k$, and let $\cal{H}$ be a $\tdc$-coloring of $G$. Let ${\cal{H}}_i$ be the restriction of the coloring $\cal{H}$ of $G$ to the component $G_i$, for $i \in [k]$. The resulting coloring ${\cal{H}}_i$ is itself a TD-coloring of $G_i$ for $i \in [k]$. Since each component of $G$ is itself a connected cograph, applying Theorem~\ref{conco} to each component of $G$, we infer that $\vert C_0^{{\cal{H}}_{i}} \vert =2$, and so each $G_i$ has two exclusive colors for all $i \in [k]$. Since there are $k$ such components, the $\tdc$-coloring $\cal{H}$ of $G$, therefore, has at least $2k$ exclusive colors.  Let $r_{\cal{H}}$ denote the maximum number of non-exclusive colors in the coloring $\cal{H}$. Thus, $\tdc(G) = \vert {\cal{H}} \vert  = 2k + r_{\cal{H}}$.

Applying Theorem~\ref{conco} to each component $G_i$ for $i \in [k]$, we have $\tdc(C_i) = \chi(C_i)$. If $\chi(C_i)< \chi(G)$ for each $i \in [k]$, then there exists a proper coloring of $G$ using less than $\chi(G)$ colors, a contradiction. Thus, there exists at least one component of $G$, say $C_j$, such that $\chi(C_j) = \chi(G)$. Therefore, the coloring ${\cal{H}}_j$ uses $\chi(G)$ colors. Among these $\chi(G)$ colors, two colors are exclusive for $C_j$ and the remaining $\chi(G)-2$ are non-exclusive colors for $C_j$, and so $r_{\cal{H}} \ge \chi(G) - 2$. However, $\chi(G) - 2$ is the maximum number of non-exclusive colors possible for any component of $G$, and so $r_{\cal{H}} \le \chi(G) - 2$. Consequently, $r_{\cal{H}} = \chi(G)-2$. Therefore, $\tdc(G)=2k+ (\chi(G)-2) =\chi(G)+2(k-1)$.  
\end{proof}

For a cograph $G$, the chromatic number $\chi(G)$ can be computed in linear-time \cite{cographs}. Thus, $\tdc(G)$ of cographs can also be computed in linear-time.

\section{Chain Graphs}
\noindent
In this section, we will show that $2 \le \tdc(G) \le 4$ for any chain graph $G$ and we characterize the chain graphs satisfying $\tdc(G) =i$ for each $i$, $2 \le i \le 4$. Below, we recall the definition and some properties of chain graphs.

A bipartite graph $G=(V,E)$ can be represented as $G = (X, Y,E)$ where $(X,Y)$ forms a partition of the vertex set $V$ of $G$ such that the sets $X$ and $Y$ are independent. Let $n_1 = \vert X\vert$ and $n_2 = \vert Y\vert$. Consider a bipartite graph $G = (X, Y,E)$, where $X=\{x_1,\ldots ,x_{n_1}\}$ and $Y=\{y_1,\ldots ,y_{n_2}\}$. ``A bipartite graph $G = (X, Y,E)$ is a \textit{chain graph} if there exists an ordering of $X$, say $(x_1,x_2, \ldots ,x_{n_1})$ such that $N(x_1) \subseteq N(x_2) \subseteq \cdots \subseteq N(x_{n_1})$''. Moreover, if $G=(X,Y,E)$ is a chain graph, then there also exists an ordering of $Y$ as well, say $(y_1,y_2, \ldots ,y_{n_2})$ such that $N(y_1) \supseteq N(y_2) \supseteq \cdots \supseteq N(y_{n_2})$. ``For a chain graph $G = (X, Y,E)$, such an ordering $(x_1,x_2, \ldots ,x_{n_1},y_1,y_2, \ldots ,y_{n_2})$ is a \textit{chain ordering} if $N(x_1) \subseteq N(x_2) \subseteq \cdots \subseteq N(x_{n_1})$ and $N(y_1) \supseteq N(y_2) \supseteq \cdots \supseteq N(y_{n_2})$". A chain ordering of given chain graph $G$ can be obtained in linear-time \cite{chain}.

Now, we define a relation $R$ on $X$ such that two vertices $x_i$ and $x_j$ in $X$ are related if $N(x_i)=N(x_j)$. Let $X_1, \ldots, X_k$ be the partition of $X$ based on the relation $R$. We now define $Y_1=N(X_1)$ and
\[
Y_i = N(X_i) \setminus \bigcup_{j=1}^{i-1} N(X_j)
\]
for $i \in [k] \setminus \{1\}$. The resulting sets $Y_1, \ldots, Y_k$ form a partition of $Y$. Such a partition $X_1, \ldots, X_k$ and $Y_1, \ldots, Y_k$ of $X$ and $Y$ is a \emph{chain partition} of $X$ and $Y$ of length~$k$. If $x \in X_i$ and $y \in Y_i$ for some $i \in [k]$, then the partition $X_1, \ldots, X_k$ of $X$ and $Y_1, \ldots, Y_k$ of $Y$ are defined so that
\[
N(x) = \bigcup_{j=1}^{i} Y_j \hspace*{0.5cm} \mbox{and} \hspace*{0.5cm} N(y) = \bigcup_{j=i}^{k} X_j.
\]

Throughout this section, we will consider an isolate-free chain graph $G$ with a chain partition $X_1, \ldots, X_k$ of $X$ and $Y_1, \ldots, Y_k$ of $Y$, respectively. Note that the number of sets in the partition of $X$ (or $Y$) is $k$.

In the following result, we establish the bounds on $\tdc(G)$ of a chain graph $G$ and we present some properties of an isolate-free chain graph.

\begin{lem}
\label{lem:chain}
If $G$ is chain graph with a chain partition of length~$k$, then the following properties hold. 
\begin{enumerate}
\item[{\rm (a)}] $\gamma_t(G)=2$.
\item[{\rm (b)}] $2 \le \tdc(G) \le 4$.
\item[{\rm (c)}] If $k \ge 2$, then $\tdc(G) \ge 3$.
\item[{\rm (d)}] If $k \ge 3$, then $\tdc(G) = 4$.
\end{enumerate}
\end{lem}
\begin{proof}
Let $G = (X, Y, E)$ be chain graph with a chain partition of length~$k$. Assume that the set $S=\{x,y\}$, where $x \in X_k$ and $y \in Y_1$. Then, $S$ is a TD-set of $G$, as a vertex $x' \in X$ totally dominates $y$ and $y' \in Y$ totally dominates $x$. So, $\gamma_t(G) \le 2$. Since $\gamma_t(F) \ge 2$ for all isolate-free graphs $F$, this yields $\gamma_t(G) = 2$. This proves part~(a).

By Corollary~\ref{planarbound}(a), we have $\gamma_t(G) \le \tdc(G) \le \gamma_t(G) + 2$. Since $\gamma_t(G) = 2$ by part~(a), this yields $2 \le \tdc(G) \le 4$, which proves part~(b).

To prove part~(c), let $k \ge 2$ and let $\cal{H}$ be a $\tdc$-coloring of $G$. Let $x \in X_1$. Let $V_1^{\cal{H}}$ be a color class totally dominated by $x$ and let the vertices in $V_1^{\cal{H}}$ be colored with color~$1$. Since $N(x) = Y_1$, the color class $V_1^{\cal{H}} \subseteq Y_1$, and so there exists a vertex in $Y_1$ with color~$1$. Let $y_1$ be such a vertex in $Y_1$ with color~$1$. We note that $y_1$ is adjacent to all vertices of $X$. Let $y \in Y_k$ and let $V_2^{\cal{H}}$ be a color class totally dominated by the vertex $y$ and let the vertices in $V_2^{\cal{H}}$ be colored with color~$2$. Since $N(y) = X_k$, the color class $V_2^{\cal{H}} \subseteq X_k$, and so there exists a vertex in $X_k$ with color~$2$. Let $x_k$ be such a vertex in $X_k$ with color~$2$. We note that $x_k$ is adjacent to whole of $Y$. Since $y$ totally dominates the color class $V_2^{\cal{H}} \subseteq X_k$, no vertex in $X_1$ is colored using color~$2$. In particular, the vertex $x \in X_1$ is not colored using color~$2$. Moreover, since the vertex $x$ is adjacent to a vertex in $Y_1$ of color~$1$, the vertex $x$ cannot be colored using color~$1$. Thus, a third color is needed to color~$x$, implying that $\tdc(G) = \vert {\cal{H}} \vert  \ge 3$. This completes the proof of part~(c).

To prove part~(d), let $k \ge 3$ and let $\cal{H}$ be a $\tdc$-coloring of $G$. We proceed exactly as in the proof of part~(c). Adopting our earlier notation in the proof of part~(c), the vertex $y_1 \in Y_1$ is colored with color~$1$ and the vertex $x_k \in X_k$ is colored with color~$2$. As observed earlier, the vertex $y_1$ is adjacent to every vertex of $X$, thus, no vertex in $X \setminus X_k$ is colored with color~$1$. Moreover, since the vertex $y \in Y_k$ totally dominates the color class $V_2^{\cal{H}} \subseteq X_k$, no vertex in $X \setminus X_k$ is colored with color~$2$. As observed earlier, the vertex $x_k$ is adjacent to every vertex of $Y$, and so no vertex in $Y \setminus Y_1$ is colored with color~$2$. Moreover, since the vertex $x \in X_1$ totally dominates the color class $V_1^{\cal{H}} \subseteq Y_1$, no vertex in $Y \setminus Y_1$ is colored with color~$1$. Hence, no vertex in $(X \setminus X_k) \cup (Y \setminus Y_1)$ is colored with color~$1$ or color~$2$. Let $x_2 \in X_2$ and let $y_2 \in Y_2$. Since $x_2$ and $y_2$ are adjacent vertices, two additional colors are therefore needed to color the vertices $x_2$ and $y_2$, and so $\tdc(G) \ge 4$. By part~(b), $\tdc(G) \le 4$. Consequently, in this case when $k \ge 3$, we have $\tdc(G) = 4$. This proves part~(d).
\end{proof}

We are now in a position to characterize the class of chain graphs for every possible value of $\tdc(G)$ in linear-time.

\begin{thm}
\label{chain:thm}
If $G$ is a chain graph with a chain partition of length~$k$, then the following properties hold. 
\begin{enumerate}
\item[{\rm (a)}] $\tdc(G) = 2$ if and only if $k=1$.
\item[{\rm (b)}] $\tdc(G) = 3$ if and only if $k=2$.
\item[{\rm (c)}] $\tdc(G) = 4$ if and only if $k \ge 3$.
\end{enumerate}
\end{thm}
\begin{proof}
Let $G = (X, Y, E)$ be an isolate-free chain graph, and let $G$ have a chain partition of length~$k$. By Lemma~\ref{lem:chain}(b), we have $2 \le \tdc(G) \le 4$. By Lemma~\ref{lem:chain}(c), if $\tdc(G) = 2$, then $k=1$. Conversely, let $k = 1$. Then, $G$ is a complete bipartite graph and every proper coloring of $G$ is a TD-coloring of $G$. Therefore, $\tdc(G) = 2$. This proves part~(a).

To prove part~(b), let $\tdc(G) = 3$. Using Lemma~\ref{lem:chain}(d) and part~(a) above, we infer that $k = 2$. To prove the converse, suppose that $k=2$. By Lemma~\ref{lem:chain}(c), we have $\tdc(G) \ge 3$. Let $\cal{H}$ be a coloring of the vertices of $G$ defined as follows. Color each vertex in $Y_1$ with color~$1$, color each vertex in $X_2$ with color~$2$, and color the vertices in $X_1 \cup Y_2$ with color~$3$. Since the resulting coloring $\cal{H}$ is a TD-coloring of $G$ using $3$ colors, and so $\tdc(G) \le  3$. Consequently, $\tdc(G) = 3$. This proves part~(b).

To prove part~(c), suppose that $\tdc(G) = 4$. By parts~(a) and~(b) above, $k \ge 3$. Conversely, if $k \ge 3$, then by using Lemma~\ref{lem:chain}(d), we have $\tdc(G) = 4$. This proves part~(c).
\end{proof}

A chain ordering of a chain graph can be obtained in linear-time \cite{chain}. A chain partition of a chain graph can also be computed in linear-time. Therefore, for a chain graph $G$, $\tdc(G)$ can also be computed in linear-time.

If $G$ is a bipartite graph, then as shown in~\cite{DC4}, $\gamma(G)\le \chi_d(G) \le \gamma(G)+2$. Further, $G$ is a complete bipartite graph if and only if $\chi_d(G)=2$~\cite{DC4}. Observe that if $G$ is a star graph $K_{1,k}$, for some $k \ge 1$, then $\gamma(G)=1$ and $\chi_d(G)=2$. Now, if $G$ is a connected chain graph different from a star graph, then using similar arguments as employed in the proofs of Lemma~\ref{lem:chain} and Theorem~\ref{chain:thm}, we remark that analogous bounds and characterizations holds for $\chi_d(G)$ of the bipartite graph $G$ as well.

\section{NP-Completeness Results}

\noindent
In this section, we study the decision version of the \textsc{Total Dominator Coloring} problem, abbreviated as the TDCD problem and we prove that the TDCD problem is NP-complete for planar graphs, connected bipartite graphs and split graphs. 
%
The following result regarding the DCD problem is known.
\begin{thm}{\rm \cite{DC2}}
\label{DC}
DCD problem is NP-complete for split graphs.
\end{thm}
For a split graph $G=(K,I,E)$ with $\vert K\vert  = \omega(G)$, here $\omega(G)$ denotes the clique number of $G$, it is known (see~\cite{DC8}) that $\omega(G) \le \chi_d(G) \le \omega(G) + 1$. We show that similar bounds also hold for $\tdc(G)$ of split graphs $G$.

\begin{lem}\label{bound}
If $G=(K,I, E)$ is a connected split graph and $\vert K\vert  = \omega(G)$, then  $\omega(G) \le \tdc(G) \le \omega(G)+1$.
\end{lem}
\begin{proof}
Let $G = (K,I,E)$ be a connected split graph and $\vert K\vert  = \omega(G)$. We note that $\tdc(G) \ge \chi(G) \ge \omega(G)$. Hence, it suffices for us to show that $\tdc(G) \le \omega(G)+1$. For this purpose, we give a TD-coloring using $\omega(G)+1$ colors. Let $\cal{H}$ be a coloring of the vertices of $G$ defined as follows. We color each vertex in $K$ with a unique color, and color the remaining vertices in $I$ with a new color. Since $I$ is an independent set, $\cal{H}$ is indeed a proper coloring of $K \cup I$. Further, $\vert {\cal{H}} \vert  = \omega(G)+1$. If $\vert K\vert =1$, then $G$ is a star graph, the vertex in the clique $K$ totally dominates the color class containing $I$, and each vertex in the set $I$ totally dominates a color class containing $K$. If $\vert K\vert  \ge 2$, then every vertex in $G$ totally dominates a color class that is contained in $K$. In both cases, $\cal{H}$ is a TD-coloring of $G$ that uses $\omega(G)+1$ colors. Hence, $\tdc(G) \le \vert {\cal{H}}\vert  = \omega(G) + 1$ and the result follows.
\end{proof}

The following corollary follows directly follows from Lemma~\ref{bound}.

\begin{co}\label{w+1}
Let $G = (K,I,E)$ be a connected split graph with $\vert K\vert  =\omega(G)$. If $\chi_d(G) = \omega(G)+1$, then $\tdc(G) = \omega(G)+1$.
\end{co}

Observe that for a star graph $G$, $\chi_d(G)=\tdc(G) =2$. Next, we show that for any split graph $G$, $\chi_d(G)$ and $\tdc(G)$ are equal.
\vspace*{-.1 cm}
\begin{lem}
\label{same}
If $G$ is a connected split graph with split partition $(K,I)$, where $\vert K\vert  =\omega(G) \ge 2$, then $\chi_d(G)=\tdc(G)$.
\end{lem}
\begin{proof}
Let $G = (K,I,E)$ be a connected split graph with $\vert K\vert  = \omega(G) \ge 2$. Let $\cal{H}$ be a $\chi_{d}$-coloring of $G$. If $\cal{H}$ is a TD-coloring of $G$, then the desired result is immediate. Now, if $\cal{H}$ is not a TD-coloring of $G$. 
Let the set $S$ contain all those vertices from $G$ that does not totally dominate any color class in the dominator coloring $\cal{H}$ of $G$.

First, we show that $K \cap S = \emptyset$. Let $v \in K \cap S$. Since $\vert K\vert  \ge 2$,
there exists $u \in K$ and $u \ne v$. Let $V_1^{\cal{H}}$ be the color class of $\cal{H}$ such that~$u \in V_1^{\cal{H}}$. Since $v \in S$, $v$ does not totally dominates $V_1^{\cal{H}}$, implying that there exists a vertex $u' \in I$ such that $u' \in V_1^{\cal{H}}$ and $v u' \notin E(G)$.
Now, $u'$ necessarily dominates some color class other than $V_1^{\cal{H}}$, say $V_2^{\cal{H}}$. Since $N(u') \subseteq K$ and $u'v\notin E(G)$, it follows that $V_2^{\cal{H}} \subset K \setminus \{v\}$ and $V_2^{\cal{H}}$ is a solitary color class. Thus, $v$ totally dominates $V_2^{\cal{H}}$, contradicting our supposition that $v \notin S$. Therefore, $K \cap S = \emptyset$. Hence, $S \subseteq I$. 

Now, let $v \in S \cap I$. Since $\cal{H}$ is a dominator coloring of $G$ and $v \in S$, the vertex $v$ dominates its own color class, say $V_1^{\cal{H}}$. Necessarily, $V_1^{\cal{H}}$ is a solitary color class and this color class is contained in $I$. To color the vertices of $K$, an additional $\omega(G)$ colors are required in the coloring $\cal{H}$. Therefore, $\chi_d(G) = \vert {\cal{H}}\vert  \ge \omega(G)+1$ and since $\chi_d(G) \le \omega(G)+1$, we get $\chi_d(G) = \omega(G) + 1$. Hence by Corollary~\ref{w+1}, $\tdc(G) = \omega(G)+1$.  
\end{proof}

By Lemma~\ref{same}, the problem of computing $\tdc(G)$ and $\chi_d(G)$ are equivalent, for a connected split graph $G$. Clearly, the TDCD problem is in NP. Now, from Theorem~\ref{DC} and Lemma~\ref{same}, we obtain the following result.
\begin{thm}\label{NPCsplit}
TDCD problem is NP-complete for split graphs.
\end{thm}
Next, we prove the NP-completeness of the TDCD problem in case of bipartite graphs. In order to do that we require the following result.
\begin{thm}{\rm \cite{bip}}
\label{approxbip}
For any graph $G$, the problem of determining a $\gamma_t$-set of $G$ cannot be approximated to within a factor of $c \ln(n)$ in polynomial time, for any constant $c<1$, unless $P=NP$. This holds true for bipartite graphs as well.
\end{thm}
From Theorem~\ref{approxbip}, it follows that it's not possible to approximate $\gamma_t(G)$ below a factor of $\ln(n)$. When $n \ge 8$, we note that $\ln(n)>2$, and so $\gamma_t(G)$ cannot be approximated within an approximation ratio of~$2$.
\begin{co}\label{2approxbip}
If $n \ge 8$, then the problem of determining a $\gamma_t$-set of $G$ cannot be approximated to within a factor of~$2$ in polynomial time, unless $P=NP$. This is true for bipartite graphs as well.
\end{co}
\begin{thm}
TDCD problem is NP-complete for connected bipartite graphs.
\end{thm}
\begin{proof}
Let $G$ be a connected bipartite graph. Clearly, TDCD problem is in NP. It remains to show that TDCD is NP-hard. On the contrary, suppose that the \textsc{Total Dominator Coloring} problem is polynomial time solvable for connected bipartite graphs. Let $\cal{H}$ be a $\tdc$-coloring of $G$ and let $C^{\cal{H}}=\{ V_1^{\cal{H}}, V_2^{\cal{H}}, \ldots, V_{\tdc(G)}^{\cal{H}}\}$ be the collection of color classes of $\cal{H}$. Now, consider the following approximation algorithm for finding a TD-set of the connected bipartite graph $G$:
\begin{algorithm}[]
\label{algor}
\textbf{Input:} A connected bipartite graph $G$.\\
\textbf{Output:} A total dominating set of $G$.\\
Compute a $\tdc$-coloring $\cal{H}$ of $G$.\\
Let $C^{\cal{H}}=\{ V_1^{\cal{H}}, V_2^{\cal{H}}, \ldots, V_{\tdc(G)}^{\cal{H}}\}$ be the collection of color classes of $\cal{H}$.\\
\For{$(i=1$ to $\tdc(G))$}{
Update $D \gets D \cup \{u_i\}$ where $u_i$ is some vertex of $V_i^{\cal{H}}$;\\
 return $D$;
}
\caption{\textbf{APPROX$\_$TDS$(G,\cal{H},C^{\cal{H}})$}}
\end{algorithm}

Note that the time complexity of \textbf{Algorithm 1} is polynomial, as the \textsc{Total Dominator Coloring} problem can be solved in polynomial time for $G$ and each step takes polynomial time. From Corollary~\ref{planarbound}(a), $\tdc(G) \le \gamma_t(G)+2$. The set $D$ obtained from \textbf{Algorithm 1} is a TD-set of cardinality~$\tdc(G) \le \gamma_t(G)+2$. As $\gamma_t(G)\ge 2$, we observe that $\gamma_t(G)+2 \le 2 \gamma_t(G)$. Thus, $D$ is a TD-set of cardinality at most $2\gamma_t(G)$. Therefore, we get a $2$-approximation algorithm for finding a TD-set of $G$, contradicting Corollary~\ref{2approxbip}. Hence, the result follows.
\end{proof}

Lastly, we consider planar graphs and we prove that the decision version of the TD-coloring problem is NP-complete in case of planar graphs using another known NP-complete problem, namely the TD-set problem. We will now formally define the decision version of the TD-set problem. The decision version of the TD-set problem, abbreviated as the TDD problem, takes a graph $G$ and a positive integer $k$ as an input and asks whether there exists a TD-set of size at most $k$. The following result is known regarding the TDD problem for planar graphs.
\begin{thm}{\rm \cite{planar}}
\label{TDplanar}
TDD problem is NP-complete for planar graphs.
\end{thm}
\begin{thm}
TDCD problem is NP-complete for planar graphs.
\end{thm}
\begin{proof}
Clearly, the TDCD problem is in NP. Next, we need to show that the TDCD problem is NP-hard. On the contrary, suppose that the TD-coloring problem is solvable in polynomial time for planar graphs. Then, we claim that total domination problem can be solved in polynomial time for planar graphs, which would contradict Theorem~\ref{TDplanar}.

Let $G$ be a planar graph. From Corollary~\ref{planarbound}(b), we have $\tdc(G) \le \gamma_t(G)+4$. Consider five copies $G_1, G_2, \ldots, G_5$ of the graph $G$, and let $G'$ be the disjoint union of these five copies of $G$. Applying Corollary~\ref{planarbound} to the graph $G'$, $\tdc(G') \le \gamma_t(G')+4$. As TDCD can be solved in polynomial time for $G'$ as well, let $\cal{H}$ be a $\tdc$-coloring of $G'$ and let $C^{\cal{H}}=\{C_1,C_2, \ldots, C_{\tdc(G')}\}$. 

Now, we define $D' = \{u_1,u_2, \ldots, u_{\tdc(G')}\}$, where $u_i \in C_i$ for $1 \le i \le \tdc(G')$. This $D'$ is a TD-set of $G'$ of cardinality at most $\gamma_t(G')+4$. Since $G'$ is union of five copies of $G$, $\gamma_t(G')= 5\gamma_t(G)$. Thus, $D'$ is a TD-set of cardinality at most $5\gamma_t(G)+4$. Let $D_i = D' \cap V(G_i)$ for $i \in [5]$, and so $\vert D' \vert = \sum_{i=1}^5 \vert D_i\vert $. Pick $j \in [5]$ such that $\vert D_j \vert  \le \vert D_i \vert  $ for all $i \in [5]$. Then, $5 \vert D_j \vert  \le \sum_{i=1}^5 \vert D_i\vert =\vert D'\vert  \le 5\gamma_t(G)+4$, implying that $\vert D_j\vert  \le \gamma_t(G)+\frac{4}{5} $. Therefore, we have a TD-set $D_j$ of $G$ of cardinality $\gamma_t(G)$. This yields a $\gamma_t$-set of $G$ in polynomial time, contradicting Theorem~\ref{TDplanar}.  
\end{proof}

\section{Conclusion}

\noindent
In this paper, we studied the \textsc{Total Dominator Coloring} problem for some important graph classes, including trees, cographs, chain graphs, split graphs, planar graphs and connected bipartite graphs. We determined the values of the total dominator chromatic number for both connected and disconnected cographs. We showed that for a chain graph $G$, $\tdc(G)$ takes one of following three values $2$, $3$ or $4$, and we characterized the chain graphs for every possible value of $\tdc(G)$. On the negative side, we have proved that the TDCD problem remains NP-complete when restricted to planar graphs, split graphs and connected bipartite graphs, strengthening the only known hardness result for the TDCD problem for general graphs. In this way, we established that the TDCD problem can not be solved in polynomial time for chordal graphs. The characterization of trees having $\tdc(T)=\gamma_t(T)+1$ was posed as an open problem in~\cite{TDC1} and we answered that by characterizing trees $T$ satisfying $\tdc(T)= \gamma_t(T) +1$. However, we remark that the condition given in our characterization cannot be checked in polynomial time. Hence, it still remains an open problem to give a polynomial time characterization of trees $T$ satisfying $\tdc(T)= \gamma_t(T) +1$. Since the total dominator chromatic number of many graph classes is still unknown, it would be interesting to work on resolving the complexity status of the \textsc{Total Dominator Coloring} problem of other important graph classes.

%
%
%
%
%

\begin{thebibliography}{10}

\bibitem{TDC7}
S.~Alikhani and N.~Ghanbari.
\newblock Total dominator chromatic number of graphs with specific
  construction.
\newblock {\em Open J. Discrete Appl. Math.}, 3(2):1--7, 2020.

\bibitem{DC8}
S.~Arumugam, J.~Bagga, and K.~R.~Chandrasekar.
\newblock On dominator colorings in graphs.
\newblock {\em Proc. Indian Acad. Sci. Math. Sci.}, 122(4):561--571, 2012.

\bibitem{DC2}
S.~Arumugam, K.~R.~Chandrasekar, N.~Misra, G.~Philip, and
  S.~Saurabh.
\newblock Algorithmic aspects of dominator colorings in graphs.
\newblock In {\em Combinatorial algorithms}, volume 7056 of {\em Lecture Notes
  in Comput. Sci.}, pages 19--30. Springer, Heidelberg, 2011.

\bibitem{planar}
R.~V. Book.
\newblock Book {R}eview: {C}omputers and intractability: {A} guide to the
  theory of {$NP$}-completeness.
\newblock {\em Bull. Amer. Math. Soc. (N.S.)}, 3(2):898--904, 1980.

\bibitem{bip}
M.~Chleb\'{\i}k and J.~Chleb\'{\i}kov\'{a}.
\newblock Approximation hardness of dominating set problems in bounded degree
  graphs.
\newblock {\em Inform. and Comput.}, 206(11):1264--1275, 2008.

\bibitem{DC4}
R.~Gera.
\newblock On the dominator colorings in bipartite graphs.
\newblock In {\em Fourth International Conference on Information Technology
  (ITNG'07)}, pages 947--952. IEEE, 2007.

\bibitem{TDC8}
N.~Ghanbari and S.~Alikhani.
\newblock More on the total dominator chromatic number of a graph.
\newblock {\em J. Inf. Optim. Sci.}, 40(1):157--169, 2019.

\bibitem{TDC10}
N.~Ghanbari and S.~Alikhani.
\newblock More on the total dominator chromatic number of a graph.
\newblock {\em J. Inf. Optim. Sci.}, 40(1):157--169, 2019.

\bibitem{HaHeHe-20}
T.~W.~Haynes, S.~T.~Hedetniemi, and M.~A.~Henning, editors.
\newblock {\em Topics in domination in graphs}, volume~64 of {\em Developments
  in Mathematics}.
\newblock Springer, Cham, 2020.

\bibitem{HaHeHe-21}
T.~W.~Haynes, S.~T.~Hedetniemi, and M.~A.~Henning, editors.
\newblock {\em Structures of domination in graphs}, volume~66 of {\em
  Developments in Mathematics}.
\newblock Springer, Cham, 2021.

\bibitem{HaHeHe-22}
T.~W.~Haynes, S.~T.~Hedetniemi, and M.~A.~Henning.
\newblock Domination in graphs: Core concepts.
\newblock {\em Manuscript (Springer, New York, 2020)}, 2022.

\bibitem{TDC12}
J.~T.~Hedetniemi, S.~M.~Hedetniemi, S.~T.~Hedetniemi, A.~A.~McRae, and D.~F.~Rall.
\newblock Total dominator partitions and colorings of graphs, february 18,
  2011.
\newblock {\em Unpublished manuscript}, 2011.

\bibitem{TDC11}
S.~M.~Hedetniemi, S.~T.~Hedetniemi, A.~A.~McRae, D.~F.~Rall, and J.~T.~Hedetniemi.
\newblock Dominator colorings of graphs, july 9, 2009.
\newblock {\em Unpublished manuscript}, 2009.

\bibitem{chain}
P.~Heggernes and D.~Kratsch.
\newblock Linear-time certifying recognition algorithms and forbidden induced
  subgraphs.
\newblock {\em Nordic J. Comput.}, 14(1-2):87--108 (2008), 2007.

\bibitem{TDC1}
M.~A.~Henning.
\newblock Total dominator colorings and total domination in graphs.
\newblock {\em Graphs Combin.}, 31(4):953--974, 2015.

\bibitem{He-21}
M.~A.~Henning.
\newblock Dominator and total dominator colorings in graphs.
\newblock In {\em Structures of Domination in Graphs}, pages 101--133.
  Springer, 2021.

\bibitem{HeYebook}
M.~A.~Henning and A.~Yeo.
\newblock {\em Total domination in graphs}.
\newblock Springer Monographs in Mathematics. Springer, New York, 2013.

\bibitem{TDC9}
P.~Jalilolghadr, A.~P.~Kazemi, and A.~Khodkar.
\newblock Total dominator coloring of circulant graphs {$C_n(a, b)$}.
\newblock {\em Util. Math.}, 115:105--117, 2020.

\bibitem{TDC5}
A.~P.~Kazemi.
\newblock Total dominator coloring in product graphs.
\newblock {\em Util. Math.}, 94:329--345, 2014.

\bibitem{TDC2}
A.~P.~Kazemi.
\newblock Total dominator chromatic number of a graph.
\newblock {\em Trans. Comb.}, 4(2):57--68, 2015.

\bibitem{TDC4}
A.~P.~Kazemi.
\newblock Total dominator chromatic number of {M}ycieleskian graphs.
\newblock {\em Util. Math.}, 103:129--137, 2017.

\bibitem{cographs}
D.~Kr\'{a}\v{l}, J.~Kratochv\'{\i}l, Z.~Tuza, and G.~J.~Woeginger.
\newblock Complexity of coloring graphs without forbidden induced subgraphs.
\newblock In {\em Graph-theoretic concepts in computer science ({B}oltenhagen,
  2001)}, volume 2204 of {\em Lecture Notes in Comput. Sci.}, pages 254--262.
  Springer, Berlin, 2001.

\bibitem{cograph}
D.~Seinsche.
\newblock On a property of the class of {$n$}-colorable graphs.
\newblock {\em J. Combinatorial Theory Ser. B}, 16:191--193, 1974.

\bibitem{Vi12}
A.~Vijayalekshmi.
\newblock Total dominator colorings in graphs.
\newblock {\em International journal of Advancements in Research and
  Technology}, 4:1--6, 2012.

\bibitem{TDC3}
A.~Vijayalekshmi.
\newblock Total dominator colorings in paths.
\newblock {\em International Journal of Mathematical Combinatorics}, 2:89--95,
  2012.

\bibitem{TDC6}
A.~Vijayalekshmi.
\newblock Total dominator colorings in caterpillars.
\newblock {\em Mathematical Combinatorics}, 2:116--121, 2014.

\end{thebibliography}
\end{document}